\documentclass[letterpaper, 10 pt, conference]{ieeeconf}
\IEEEoverridecommandlockouts
\usepackage{amsmath,dsfont}
\usepackage{amsthm}
\usepackage{graphicx}
\usepackage{amsfonts}
\usepackage{float}
\usepackage{calrsfs}
\DeclareMathAlphabet{\pazocal}{OMS}{zplm}{m}{n}
\DeclareMathOperator{\epi}{\mathbf{epi}}

\usepackage{algorithmicx}
\usepackage{algorithm}
\usepackage{algpseudocode}
\usepackage{lipsum}
\usepackage{subfig}
\usepackage{amssymb}
\usepackage{bbm}
\usepackage{float}
\usepackage{balance}
\theoremstyle{plain}
\newtheorem{thm}{\textbf{Theorem}}
\newtheorem{lem}{\textbf{Lemma}}
\newtheorem{prop}{\textbf{Proposition}}
\newtheorem{cor}{\textbf{Corollary}}

\theoremstyle{definition}
\newtheorem{defn}{\textbf{Definition}}

\newtheorem{exmp}{\textbf{Example}}
\newtheorem{ass}{\textbf{Assumption}}
\theoremstyle{remark}
\newtheorem{rem}{\textbf{Remark}}

\usepackage{algpseudocode}
\algdef{SE}[DOWHILE]{Do}{DoWhile}{\algorithmicdo}[1]{\algorithmicwhile\ #1}%

\begin{document}

\title{\LARGE \bf
Compositional Set Invariance in Network Systems with Assume-Guarantee Contracts
}
\author{Yuxiao Chen, James Anderson, Karan Kalsi, Steven H. Low, and Aaron D. Ames
\thanks{Yuxiao Chen and Aaron D. Ames are with the Department of Mechanical and Civil Engineering, Caltech,
        Pasadena, CA, 91106, USA. Emails:
        {\tt\small \{chenyx, ames\}@caltech.edu}}
\thanks{James Anderson and Steven H. Low are with the Computing and Mathematical Sciences Department, Caltech,
         Pasadena, CA, 91106, USA. Emails:
        {\tt\small \{james, slow\}@caltech.edu}}
\thanks{Karan Kalsi is with Pacific Northwest National Laboratory, Richland, WA, 99352, USA. Email:
        {\tt\small Karanjit.Kalsi@pnnl.gov}}
\thanks{This work is supported by the Battelle Memorial Institute, Pacific Northwest Division, Grant \#424858.}
}

%\author{\IEEEauthorblockN{Yuxiao Chen}
%\IEEEauthorblockA{\textit{Dept. of Mechanical and Civil Engineering} \\
%\textit{Caltech}\\
%Pasadena, CA \\
%chenyx@caltech.edu}
%\and
%\IEEEauthorblockN{James Anderson}
%\IEEEauthorblockA{\textit{Computing and Mathematical Sciences} \\
%\textit{Caltech}\\
%Pasadena, CA \\
%james@caltech.edu}
%\and
%\IEEEauthorblockN{Karan Kalsi}
%\IEEEauthorblockA{\textit{Pacific Northwest National Laboratory}\\
%Richland, WA \\
%Karanjit.Kalsi@pnnl.gov}
%\and
%\IEEEauthorblockN{Steven Low}
%\IEEEauthorblockA{\textit{Computing and Mathematical Sciences} \\
%\textit{Caltech}\\
%Pasadena, CA \\
%slow@caltech.edu}
%\and
%\IEEEauthorblockN{Aaron Ames}
%\IEEEauthorblockA{\textit{Dept. of Mechanical and Civil Engineering} \\
%\textit{Caltech}\\
%Pasadena, CA \\
%ames@caltech.edu}
%}

\maketitle
\begin{abstract}
  This paper presents an assume-guarantee reasoning approach to the computation of robust invariant sets for network systems. Parameterized signal temporal logic (pSTL) is used to formally describe the behaviors of the subsystems, which we use as the template for the contract. We show that set invariance can be proved with a valid assume-guarantee contract by reasoning about individual subsystems. If a valid assume-guarantee contract with monotonic pSTL template is known, it can be further refined by value iteration. When such a contract is not known, an epigraph method is proposed to solve for a contract that is valid, ---an approach that has linear complexity for a sparse network. A microgrid example is used to demonstrate the proposed method. The simulation result shows that together with control barrier functions, the states of all the subsystems can be bounded inside the individual robust invariant sets.
\end{abstract}
\section{Introduction}\label{sec:intro}
Correct-by-construction control synthesis has seen recent success in safety-critical applications such as vehicle control \cite{nilsson2014preliminary,chen2018validating} and robot navigation \cite{chen2018obstacle}. This approach bases the controller on concepts such as reachable set and control invariant sets to synthesize a controller that is capable of enforcing safety. However, reachability analysis and invariant set computation rely on computational tools such as Hamilton Jacobi \cite{mitchell2005time}, Linear Matrix Inequality (LMI) \cite{khlebnikov2011optimization} and sum of squares (SOS) programming \cite{papachristodoulou2005tutorial,prajna2004nonlinear} ---these methods scale poorly with the dimension of the system. Because of this limitation, sometimes referred to as ``the curse of dimensionality,'' the applications of the correct-by-construction control synthesis have been limited to systems with low state dimension. There has been effort to break ``the curse of dimensionality,'' which typically utilizes either the compositional analysis or system symmetry \cite{hussien2017abstracting,smith2016interdependence,nilsson2016control,anderson2011model}. For example, in \cite{smith2016interdependence}, the weakly coupled longitudinal and lateral dynamics of the vehicle are treated independently by finding a bound on the coupling effect. In \cite{nilsson2016control}, when a large network system consists of small subsystems that are identical, the symmetry is utilized to compute invariant sets for a large number of subsystems. However, correct-by-construction synthesis for network systems with heterogeneous subsystems and strong coupling between them remains an open problem. One example is the power grid, which consists of various types of generation buses and load buses, as shown in Fig. \ref{fig:grid_network}.
\begin{figure}
  \centering
  \includegraphics[width=3in]{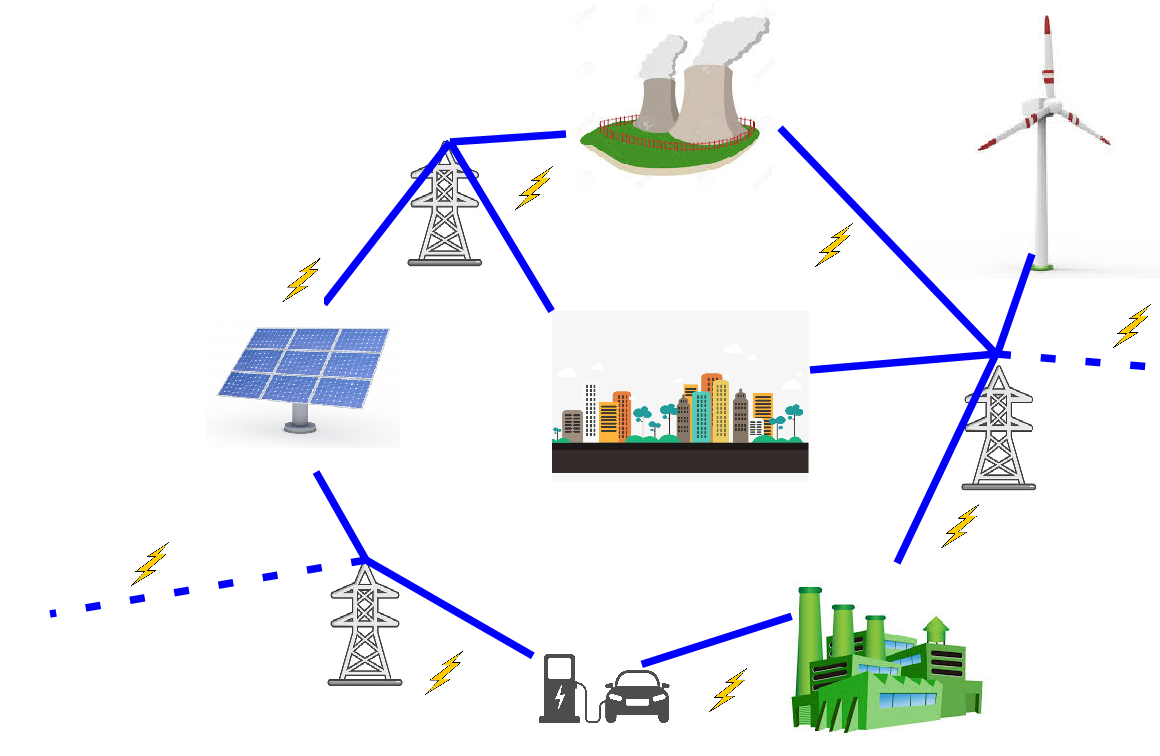}
  \caption{Power grid with generator buses and load buses}\label{fig:grid_network}
\end{figure}

One promising direction is the assume-guarantee contract \cite{alur1999reactive}, which decomposes the overall performance guarantee into individual contracts for each subsystem. Every subsystem in the network takes the performance guarantee from other subsystems as assumptions and in turn gives its own performance guarantee, which then becomes part of the assumptions for other subsystems in the network. For discrete transition systems, there exists algorithms that automatically generate assume-guarantee contracts \cite{bobaru2008automated}. However, for dynamic systems with continuous state space, to the knowledge of the authors, there exists no method that generates a valid assume-guarantee contract automatically and the design of such a contract depends on engineering intuition and trial-and-error. This is the problem that will be studied in this paper.

In this paper, we propose the epigraph algorithm that searches for valid assume-guarantee contracts for a network system via optimization ---this leads to robust invariant sets for the network system. The proposed method is compatible with any existing method for invariant set computation and enjoys linear complexity for sparse networks. The epigraph method consists of three steps. First, for each subsystem, we associate it with a parameterized assume-guarantee contract and define a local function, $\lambda_i$, that characterizes the relationship between the assumption parameters and the guarantee parameters. Then a grid sampling algorithm is used to compute an inner approximation of the epigraph of $\lambda_i$ for each subsystem, denoted by $\epi(\lambda_i)$. Finally, given $\epi(\lambda_i)$, a centralized optimization solves for a set of parameters that makes the overall assume-guarantee contract valid.

In the remainder of the paper, Section \ref{sec:setup} presents the problem setup and reviews some fundamental tools and concepts; Section \ref{sec:set_invariance} presents the main result of proving set invariance with assume-guarantee reasoning for network systems; Section \ref{sec:epigraph} presents the epigraph algorithm that searches for a valid assume-guarantee contract with optimization; the proposed method is demonstrated with an example of microgrid control in Section \ref{sec:application} and finally we conclude in Section \ref{sec:conclusion}.
\section{Problem setup}\label{sec:setup}
In this section, we present the problem setup and some fundamental concepts and tools.

\textit{Nomenclature}  For the remainder of the paper, $\mathbb{N}$ denotes the set of natural numbers, $\mathbb{R}$ denotes the set of real numbers,  $\mathbb{B}=\{0,1\}$ denotes the set of binary numbers. We use $p\in\pazocal{P}$ to denote a parameter, with $\pazocal{P}$ as its domain. For a variable $x\in\pazocal{X}$, $x(t)$ denotes its value at the $t$-th time instance, $x(\cdot)$ denotes the evolution trajectory of $x$ for $t=0,1,...$ Correspondingly, $\pazocal{X}(\cdot)$ denotes the space of all possible evolutions of $x$. To avoid confusion, in a value iteration process, $p[i]$ denotes the value of a parameter $p$ after the $i$-th iteration.
\subsection{Network dynamics}
We consider a network dynamic system that is decomposed into subsystems with an assume-guarantee contract and treats the coupling between neighboring subsystems as bounded disturbance. Therefore, the following product of subsystems is considered:
\footnote{Note that for a general network dynamical system, the corresponding model would be defined over a graph structure \cite{sandell1978survey};  as noted, in the context of this paper, because we view the coupling between systems via bounded disturbances, we can consider a network of dynamical systems as simply the product system.}

\begin{equation}\label{eq:network_system}
  \Sigma = \Sigma_1 \times \Sigma_2 \times ... \times \Sigma_N.
\end{equation}
It is assumed that for each subsystem, there exists an output vector and all the coupling between subsystems are through the outputs of the discrete-time subsystems:
\begin{equation}\label{eq:dynamic_equation}
 \begin{array}{l}
x_i^+ = {f_i}\left( {{x_i},{y_{\pazocal{N}_i}},{u_i},{d_i}} \right),\\
{y_i} = {h_i}\left( {{x_i}} \right),i = 1,2,...,N,
\end{array}
\end{equation}
where $x_i\in\pazocal{X}_i\subseteq\mathbb{R}^{n_i}$, $u_i\in\pazocal{U}_i\subseteq\mathbb{R}^{m_i}$, $d_i\in\pazocal{D}_i\subseteq\mathbb{R}^{l_i}$, $y_i\in\pazocal{Y}_i\subseteq\mathbb{R}^{s_i}$ are the state, control input, exogenous disturbance input and output of $\Sigma_i$. $y_{\pazocal{N}_i}$ are the outputs of the neighboring subsystems of $\Sigma_i$. We use $\pazocal{N}_i$ to denote the indices of $\Sigma_i$'s neighboring subsystems:
\begin{equation}\label{eq:NB}
  y_{\pazocal{N}_i} = [y_{j_1}^\intercal,y_{j_2}^\intercal,...,y_{j_{N_i}}^\intercal]^\intercal, j_1,j_2,...j_{N_i}\in \pazocal{N}_i, \left| {\pazocal{N}_i} \right| = {N_i}.
\end{equation}
We denote the overall state space and output space as $\pazocal{X}=\pazocal{X}_1 \times ...\times \pazocal{X}_N$ and $\pazocal{Y}=\pazocal{Y}_1\times ...\times \pazocal{Y}_N$, respectively. Without loss of generality, it is assumed that for all $\Sigma_i$, the operating point is the origin and
\begin{equation}
  h_i(0)=0.
\end{equation}
The behavior $y_i(\cdot)$ is completely determined by $x_i(0)$, $y_{\pazocal{N}_i}(\cdot)$, $u_i(\cdot)$ and $d_i(\cdot)$, let $\pazocal{I}_i(\cdot)=\pazocal{X}_i\times \pazocal{Y}_{\pazocal{N}_i}(\cdot) \times \pazocal{U}_i(\cdot) \times \pazocal{D}_i(\cdot)$ denote the space of input trajectories and initial conditions of the system $\Sigma_i$ and $\pazocal{Y}_i(\cdot)$ is the space of all possible output trajectories of $\Sigma_i$. A dynamic system $\Sigma_i\subseteq2^{\pazocal{I}_i(\cdot)}\times 2^{\pazocal{Y}_i(\cdot)}$ is understood as a subset of possible input and output behavior pairs.
\subsection{Parameterized Signal Temporal Logic}\label{sec:pSTL}
The approach we take in this paper is to use assume-guarantee contract to resolve ``the curse of dimensionality'' for large network systems. This work differs from the assume-guarantee approach used for verification and synthesis of transition systems \cite{alur1999reactive,puasuareanu2008learning}, where the contract appears as a set of admissible states or actions, we present an assume-guarantee approach for dynamic systems with continuous input and state spaces. Parametric Signal Temporal Logic \cite{asarin2011parametric,bombara2016decision} is used to formally assess the behaviors of the systems, which is used as the template for specifications. A Signal Temporal Logic (STL) formula $\phi:\pazocal{X}(\cdot)\to\mathbb{B}$ is written using the following grammar:
\begin{equation}\label{eq:pSTL_format}
  \phi = \top \mid \mu \mid \lnot \phi \mid \phi_1\wedge\phi_2\mid\phi_1\mathbf{U}_I \phi_2,
\end{equation}
where $\top$ is the logic tautology, $\mu:\pazocal{X}\to\mathbb{B}$ is a logic proposition, $\lnot$ is Boolean negation and $\wedge$ is a Boolean AND, $I$ is an interval. The semantics are formally given as follows:
\begin{table}[H]
\centering
\begin{tabular}{lll}
$(\mathbf{x},t)\models \mu$                             &iff & $\mathbf{x}$ satisfies $\mu$ at time $t$\\
$(\mathbf{x},t)\models \lnot\phi$                      &iff & $(\mathbf{x},t)\not\models \phi$\\
$(\mathbf{x},t)\models \phi_1\wedge\phi_2$             &iff & $(\mathbf{x},t)\models \phi_1$ and $(\mathbf{x},t)\models \phi_2$\\
$(\mathbf{x},t)\models \phi_1\mathbf{U}_{[a,b]}\phi_2$ &iff & {\begin{tabular}{l}
                                                                $\exists t'\in t+[a,b]$ s.t. $(\mathbf{x},t')\models\phi_2$ \\
                                                                and $\forall t''\in[t,t'], (\mathbf{x},t'')\models\phi_1$
                                                              \end{tabular} }
\end{tabular}
\end{table}

From the above basic grammar, one can derive additional temporal operators $\lozenge_I \phi = \top \mathbf{U}_I \phi$, which means ``$\phi$ is eventually true during $I$,'' and $\square_I\phi = \lnot(\lozenge_I \lnot\phi)$, which means ``$\phi$ is always true in $I$''. When $I$ is not specified, it is assumed that by default $I=[0,\infty)$.
\begin{rem}
  A pSTL is extended to discrete-time signals by considering the sampling instances, as discussed in \cite{fainekos2009robustness}.
\end{rem}

For a STL formula $\phi$, $L(\phi) = \left\{x(\cdot)\in\pazocal{X}(\cdot)\mid x(\cdot) \models \phi\right\}$ is the language of the formula. A partial order is defined among the STL formulas as $\phi_1 \preceq \phi_2$ if $\forall x(\cdot)\in \pazocal{X}(\cdot), (x(\cdot)\models \phi_1) \Rightarrow (x(\cdot) \models \phi_2)$, or equivalently, $L(\phi_1)\subseteq L(\phi_2)$.

A pSTL formula is a STL formula with parameters. For example, $\phi = \square_{[a,b]}(x\ge c)$ can be represented as the following pSTL: $\varphi(a,b,c) = \square_{[a,b]}(x\ge c)$, where $a,b$ and $c$ are the parameters and $\varphi:\mathbb{R}^3\to\mathbb{B}$ is the pSTL template. For the rest of the paper, it is assumed that all the pSTL formulas are defined on partially ordered parameter domains. Given a parameter domain $\pazocal{P}$, the partial order is denoted as $\le_{\pazocal{P}}$.
\begin{defn}

A pSTL formula $\varphi(p)$ is \textit{monotonically increasing} if
\begin{equation}\label{eq:mono_in}
  \forall p_1,p_2\in \pazocal{P},\quad p_1 \le_{\pazocal{P}} p_2\Rightarrow \varphi(p_1)\preceq \varphi(p_2),
\end{equation}
and \textit{monotonically decreasing} vice versa.
\end{defn}
For example, $\varphi(p)=\lozenge_{[0,p]}(x\ge0)$ is monotonically increasing and $\varphi(p)=\square_{[0,\infty)}(x\ge p)$ is monotonically decreasing.

For a pSTL $\varphi:\pazocal{P}_1\to\mathbb{B}$, if $\pazocal{P}_1\subseteq \pazocal{P}_2$, then $\forall p\in\pazocal{P}_2$, $\varphi(p) = \varphi(p_\downarrow\pazocal{P}_1)$, where $\downarrow$ denotes the projection of $p$ onto $\pazocal{P}_1$.

\subsection{Assume-Guarantee Contract for Network Systems}\label{sec:network_ag}

Next, we present a framework that gives performance guarantee to the network system based on assume-guarantee reasoning. First, the definition of assume-guarantee contract is formally defined, which is adopted from \cite{kim2017small}.
\begin{defn}[Assume-Guarantee Contract]
An assume-guarantee contract $\pazocal{C}$ for the dynamic system $\Sigma$ is a pair $[\phi_a,\phi_g]$ consisting of an assumption $\phi_a$ and a guarantee $\phi_g$ that encode the requirement that the logical implication $\phi_a \Rightarrow\phi_g$ holds.
\end{defn}
An assume-guarantee contract $\pazocal{C}=[\phi_a,\phi_g]$ is true for a dynamic system $\Sigma$ if $\Sigma \cap L(\phi_a)\subseteq L(\phi_g)$, or written compactly as $\phi_a\wedge\Sigma\preceq\phi_g$ with a slight abuse of notation.

\begin{defn}[Parameterized Assume-Guarantee Contract]
An assume-guarantee contract $\pazocal{C}=[\phi_a,\phi_g]$ is in parameterized form if there exists a pSTL $\phi_a = \varphi_a(p_a)$, a pSTL $\phi_g = \varphi_g(p_g)$ and a mapping $\lambda:\pazocal{P}_a\to \pazocal{P}_g$ such that $\pazocal{C}(p_a)=[\varphi_a(p_a),\varphi_g(\lambda(p_a))]$.
\end{defn}

In particular, $\phi_a$ consists of two parts:
\begin{equation}\label{eq:assume_partition}
  {\phi _a} =\phi_{ae}\wedge\phi_{af} = {\varphi _{ae}}\left( {{p_{ae}}} \right) \wedge {\varphi _{af}}\left( {{p_{af}}} \right),
\end{equation}
where $\phi _{ae}$ is the specification for exogenous environment behavior and $\phi _{af}$ is the feedback specification.
\begin{defn}[Parameterized Network Assume-Guarantee Contract]
  For a network system defined in \eqref{eq:network_system}, a parameterized network assume-guarantee contract consists of individual parameterized assume-guarantee contracts $\pazocal{C}_i$ for each subsystem $\Sigma_i$. Each subcontract $\pazocal{C}_i$ consists of $\phi_a^i=\varphi_{ae}^i(p_{ae}^i)\wedge\varphi_{af}^i(p_{af}^i)$ and $\phi_g^i=\varphi_g^i(p_g^i)$. Denote $p_{ae}=\bigcup\limits_{i=1}^{N}{p_{ae}^i}$, $p_{af}=\bigcup\limits_{i=1}^{N}{p_{af}^i}$ and $p_g=\bigcup\limits_{i=1}^{N}{p_g^i}$ as the overall parameters for the environment specification, feedback specification and guarantee specification, with corresponding domain $\pazocal{P}_{ae}$, $\pazocal{P}_{af}$ and $\pazocal{P}_g$. For a specific subsystem, $\varphi_{af}^i(p_{af})=\varphi_{af}^i(p_{af}\downarrow\pazocal{P}_{af}^i)$ and the same for $\varphi_{ae}$ and $\varphi_g$.
  For simplicity of notation, let
  \begin{equation}
  \begin{aligned}
    \phi_{ae} &=&\varphi_{ae}(p_{ae}) &=&\bigwedge\limits_{i=1}^N{\phi_{ae}^i} &=&\bigwedge\limits_{i=1}^N{\varphi_{ae}^i}(p_{ae}^i) \\
    \phi_{af} &=&\varphi_{af}(p_{af}) &=&\bigwedge\limits_{i=1}^N{\phi_{af}^i} &=&\bigwedge\limits_{i=1}^N{\varphi_{af}^i}(p_{af}^i) \\
    \phi_g    &=&\varphi_g(p_g)       &=&\bigwedge\limits_{i=1}^N{\phi_g^i}    &=&\bigwedge\limits_{i=1}^N{\varphi_g^i}(p_g^i) .
  \end{aligned}
  \end{equation}
\end{defn}
\begin{rem}
  Note that some parameters may appear in more than one subcontract, the overall parameters $p_{ae}$, $p_{af}$ and $p_g$ remove the repetition.
\end{rem}
%We consider a network system setting as shown in \eqref{eq:dynamic_equation} with a parameterized network assume-guarantee contract, where $\phi_{ae}^i:\pazocal{X}_i\times\pazocal{U}_i(\cdot)\times\pazocal{D}_i(\cdot)\to\mathbb{B}$ specifies the behavior of the environment including the exogenous disturbance, feedback control and initial condition for $\Sigma_i$ and $\phi_{af}^i:\pazocal{Y}_{\pazocal{N}_i}(\cdot)\to\mathbb{B}$ specifies the behavior of neighboring subsystem outputs.
%Note that $\phi_{ae}^i$ is local to $\Sigma_i$, while $\phi_{af}^i$ is affected by the other subsystems and may even overlap with $\phi_{af}^j$ of another subsystem. We pose the following consistency assumption:
%\begin{ass}\label{ass:consistency}
%For all $p\in\pazocal{P}_{af}$, the overall behavior allowed by the assumption is not empty, i.e.,
%\begin{equation}\label{eq:consitency}
%  \forall p\in\pazocal{P}_{af}, L\left(\bigwedge\limits_{i=1}^{N} {\left(\varphi_{af}^i(p\downarrow \pazocal{P}_{af}^i)\wedge \varphi_{ae}^i\right)}\right) \neq \emptyset,
%\end{equation}
%\end{ass}
%A counterexample that violates Assumption \ref{ass:consistency} is given as follows. Suppose $\Sigma_1$ is the neighbor of both $\Sigma_2$ and $\Sigma_3$, then $\phi_{af}^2=\square (y_1\ge 1)$, $\phi_{af}^3=\square (y_1\le 0)$, violate the consistency assumption.
\section{Set invariance with assume-guarantee contract}\label{sec:set_invariance}
We now present the main result of this paper, which utilize assume-guarantee reasoning to prove set invariance for network systems.
\begin{thm}[Assume-guarantee reasoning]\label{thm:ag}
Consider the network system in \eqref{eq:dynamic_equation} associated with a parameterized network assume-guarantee contract. Suppose the following are satisfied:
\begin{itemize}
  \item[] 1.
  Each subsystem satisfies a subcontract $\pazocal{C}_i(p_a^i)$, that is, $\forall p_a^i \in \pazocal{P}_a^i,{\Sigma _i} \wedge \phi _a^i\preceq \phi _g^i$, where
  \begin{equation}\label{eq:local_ag_ass}
    \begin{aligned}
\phi _a^i &= \varphi _{ae}^i\left( {p_{ae}^i} \right) \wedge \varphi _{af}^i\left( {p_{af}^i} \right),\\
\phi _g^i &= \varphi _g^i\left( {{\lambda _i}\left( {p_{ae}^i,p_{af}^i} \right)} \right).
\end{aligned}
  \end{equation}
  \item[] 2. %The feedback parameters spaces $\pazocal{P}_{af}=\bigcup\limits_{i=1}^{N} {\pazocal{P}_{af}^i}$ and guarantee parameter space $\pazocal{P}_{g}=\bigcup\limits_{i=1}^{N} {\pazocal{P}_{g}^i}$ satisfy
%\begin{equation}\label{eq:recursive_reason}
%\begin{array}{c}
%  \pazocal{P}_{af}\subseteq\pazocal{P}_g,\\
%  \forall p\in\pazocal{P}_g,\varphi_g(p)\Rightarrow\varphi_{af}(p), \\
%  \varphi_g(p)=\bigwedge\limits_{i=1}^{N} {\varphi_g^i(p)}, \\
%  \varphi_{af}(p)=\bigwedge\limits_{i=1}^{N} {\varphi_{af}^i(p)}.
%\end{array}
%\end{equation}
%Let $p_g=[p_g^1,p_g^2,...p_g^N]^\intercal$ and ${p_g} = \hat \Lambda ({p_{af}}) = {\left[ {{{\hat \lambda }_1}{{\left( {p_{af}^1} \right)}^\intercal},{{\hat \lambda }_N}{{\left( {p_{af}^N} \right)}^\intercal}} \right]^\intercal}$.
There exists a mapping $\Gamma:\pazocal{P}_g\to\pazocal{P}_{af}$ such that
\begin{equation}\label{eq:recursive_reasoning}
  \varphi_{af}^i(\gamma_i(p_g))\preceq \varphi_g(p_g),
\end{equation}
where $\gamma_i(p_g)=\Gamma(p_g)\downarrow\pazocal{P}_{af}^i$.
  \item[] 3. There exists environment parameters $p_{ae}\in\pazocal{P}_{ae}$ such that $\varphi_{ae}(p_{ae})$ is satisfied.
  \item[] 4. There exists feedback parameters $p_{af}[0]\in\pazocal{P}_{af}$ that $\varphi_{af}(p_{af}[0])$ is true.
\end{itemize}
 Given $p_{ae}^i$, define $\hat{\lambda}_i(\cdot)=\lambda_i(p_{ae}^i,\cdot)$.
Let
\begin{equation}\label{eq:hat_Lambda}
  \hat{\Lambda}(p_{af})=[\hat{\lambda}_1(p_{af}^1)^\intercal,\ \hat{\lambda}_2(p_{af}^2)^\intercal,\ ...\ \hat{\lambda}_N(p_{af}^N)^\intercal]^\intercal,
\end{equation}
then define recursively
\begin{equation}\label{eq:recursion}
\begin{aligned}
  p_g[k] &= \hat{\Lambda}(p_{af}[k]) \\
  p_{af}[k+1]&=\Gamma(p_g[k]).
\end{aligned}
\end{equation}
Under these conditions, the network system satisfies
\begin{equation}\label{eq:res_guarantee}
  \hat{\phi}_g=\bigwedge\limits_{k=0}^{\infty}{\varphi_g(p_g[k])}.
\end{equation}
\end{thm}
\begin{proof}
  Since $p_{ae}^i$ and $p_{af}^i[0]$ exists so that $\phi_{ae}^i$ and $\phi_{af}^i[0]$ are satisfied, we can build the following infinite sequence of pSTL that the network system satisfies from \eqref{eq:local_ag_ass}, \eqref{eq:recursive_reasoning} and \eqref{eq:recursion}:
  \begin{equation}\label{eq:seq_pSTL}
  \begin{array}{l}
    \bigwedge\limits_{i=1}^{N}{\varphi_{ae}^i(p_{ae}^i)} \wedge \bigwedge\limits_{i=1}^{N}{\varphi_{af}^i(p_{af}^i[0])} \wedge\\
    \left( \bigwedge\limits_{i=1}^{N}{\varphi_{ae}^i(p_{ae}^i)} \wedge \bigwedge\limits_{i=1}^{N}{\varphi_{af}^i(p_{af}^i[0])} \Rightarrow \bigwedge\limits_{i=1}^{N}{\varphi_g^i(p_g^i[0])}\right)\wedge \\
    \left(\bigwedge\limits_{i=1}^{N}{\varphi_g^i(p_g^i[0])} \Rightarrow \bigwedge\limits_{i=1}^{N}{\varphi_{af}^i(p_{af}^i[1])} \right)\wedge\\
    ...
  \end{array}
\end{equation}
which implies \eqref{eq:res_guarantee}.
\end{proof}
Next, we use Theorem \ref{thm:ag} to show set invariance of a network system using assume-guarantee reasoning. First, we give the definition of a robust control invariant set.
\begin{defn}
For the dynamic system described in \eqref{eq:dynamic_equation}, given $\pazocal{U}_i$, $\pazocal{D}_i$ and $y_{\pazocal{N}_i}^{\max}$, a set $\pazocal{S}_i\subseteq\mathbb{R}^{n_i}$ is \textit{robust control invariant} if
\begin{equation}\label{eq:RCI}
\begin{array}{c}
\forall x_i\in \pazocal{S}_i,\forall d_i\in\pazocal{D}_i,\forall \left|y_{\pazocal{N}_i}\right|\le y_{\pazocal{N}_i}^{\max},\exists u_i\in\pazocal{U}_i \\
s.t.\quad x_i^+ = f_i(x_i,y_{\pazocal{N}_i}.u_i,d_i)\in \pazocal{S}_i.
\end{array}
\end{equation}
\end{defn}
\begin{thm}[Set invariance of a network system with assume-guarantee contract]\label{thm:set_invariance}
  Consider the network system described in \eqref{eq:dynamic_equation}, suppose that all $y_i$ are scalars and there exists a feedback controller $u_i=k(x_i,y_{\pazocal{N}_i},d_i)$ such that for a given bound on $\left|y_{\pazocal{N}_i}\right|\le y_{\pazocal{N}_i}^{\max}$, a given bound $\pazocal{D}_i$ of $d_i$ and a given set $\pazocal{S}_i$ of $x_i$, the following is true:
  \begin{equation}\label{eq:invariance_a1}
  \begin{array}{c}
    \forall x_i\in \pazocal{S}_i,\quad \forall d_i\in\pazocal{D}_i,\quad \forall \left|y_{\pazocal{N}_i}\right|\le y_{\pazocal{N}_i}^{\max}, \\
    x_i^+={f_i}\left( {{x_i},{y_{\pazocal{N}_i}},{u_i},{d_i}} \right)\in \pazocal{S}_i,
  \end{array}
  \end{equation}
  \begin{equation}\label{eq:invariance_a2}
    \max\limits_{x_i\in\pazocal{S}_i}{\left|h_i(x)\right|}\le y_i^{\max},
  \end{equation}
  where $y_{\pazocal{N}_i}^{\max}$ is a projection of $y^{\max}$ onto $\pazocal{Y}_{\pazocal{N}_i}$.
  Then
  \begin{equation}\label{eq:inv_thm_ag}
  \begin{array}{c}
    \bigwedge\limits_{i=1}^{N} {\left(x_i(0)\in\pazocal{S}_i\wedge \square ( u_i=k_i(x_i,y_{\pazocal{N}_i},d_i))\wedge \square ( d_i\in\pazocal{D}_i)\right)} \\
    \Rightarrow\bigwedge\limits_{i=1}^{N}{\square (x_i\in \pazocal{S}_i)},
  \end{array}
  \end{equation}
   that is, $\pazocal{S}_1 \times \pazocal{S}_2 \times ...\times \pazocal{S}_N$ is robust control invariant.
\end{thm}
\begin{proof}
  Let
  \begin{equation}\label{eq:inv_ae}\
  \begin{array}{c}
    \phi_{ae}^i=(x_i(0)\in\pazocal{S}_i)\wedge  \square\left(d_i\in\pazocal{D}_i\right)\\
    \wedge\square\left(u_i=k(x_i,y_{\pazocal{N}_i},d_i)\right),
  \end{array}
  \end{equation}
\begin{equation}\label{eq:inv_af}
  \phi_{af}^i = \varphi_{af}^i(T_i) = \square_{[0,T_i]}\left|y_{\pazocal{N}_i}\right|\le y_{\pazocal{N}_i}^{\max},
\end{equation}
\begin{equation}\label{eq:inv_g}
  \phi_{g}^i = \varphi_{g}^i(T_i) = \square_{[0,T_i]}{x_i\in\pazocal{S}_i};
\end{equation}
and let $\hat{\lambda}_i(T_i)=T_i+T_s$, $\Gamma(T)=T$, where $T_s$ is the time step of the discrete dynamics in \eqref{eq:dynamic_equation}, $T=[T_1,T_2,...,T_N]^\intercal$.

Among the 4 assumptions of Theorem \ref{thm:ag}, Assumption 1 is satisfied due to \eqref{eq:invariance_a1}, Assumption 2 is satisfied by \eqref{eq:invariance_a2} with $\Gamma$ defined above. Assumption 3 is satisfied by \eqref{eq:inv_ae} and Assumption 4 is satisfied by setting $T_i=0$ for all $i$ in \eqref{eq:inv_af}. Then, by Theorem \ref{thm:ag}, the guarantee for the network system is
\begin{equation}\label{eq:inv_thm_g}
  \hat{\phi}_g^i=\bigwedge\limits_{k=0}^{\infty}{\square_{[0,k\cdot T_s]}{x_i\in\pazocal{S}_i}},
\end{equation}
which is simplified to
\begin{equation}\label{eq:inv_thm_g1}
  \forall i=1,...,N, \square_{[0,\infty)}{x_i\in\pazocal{S}_i}.
\end{equation}
\end{proof}

\begin{lem}\label{lem:check_out}
  Consider the following assume-guarantee contract for a subsystem $\Sigma_i$:
  \begin{equation}\label{eq:inv_ag}
    \begin{aligned}
\varphi _a^i\left( {{y_{\pazocal{N}_i}^{\max }}} \right)&: =\square \left| {{y_{\pazocal{N}_i}}} \right| \le y_{\pazocal{N}_i}^{\max }\\
\varphi _g^i\left( {{{\bar y}_i^{\max }}} \right)&: =\square \left| {{y_i}} \right| \le \bar y_i^{\max }
\end{aligned}
\end{equation}
Suppose there exist monotonically increasing functions $\lambda_i$ such that
\begin{equation}
  \forall i=1,...N,\forall y_{\pazocal{N}_i}^{\max }\ge0, \varphi _a^i\left( {{y_{\pazocal{N}_i}^{\max }}} \right) \to \varphi _g^i\left( {{\lambda _i}\left( {{y_{\pazocal{N}_i}^{\max }}} \right)} \right)
\end{equation}
  and
  \begin{equation}\label{eq:check_out0}
    \exists {y^{\max }}[0]\; s.t.\;\;\forall i = 1,...,N,\;{\lambda _i}\left( {y_{\pazocal{N}{_i}}^{\max }}[0] \right) \le y_i^{\max }[0]
  \end{equation}
  then the network system satisfies
  \begin{equation}
    \left| {y(0)} \right| \le {{ y}^{\max }[0]}\quad \Rightarrow \quad\square\left| y \right| \le \Lambda(y^{\max}[0]),
  \end{equation}
\\
  where $\Lambda ({y^{\max }}) = {\left[ {{\lambda _1}\left( {y_{\pazocal{N}{_1}}^{\max }} \right),...,{\lambda _N}\left( {y_{\pazocal{N}{_N}}^{\max }} \right)} \right]^\intercal}$
\end{lem}
The proof follows similar reasoning as Theorem \ref{thm:set_invariance} and is omitted here.

The condition in \eqref{eq:check_out0} is referred to as the validity condition, which is crucial to our assume-guarantee approach of computing invariant sets for network systems.

Lemma \ref{lem:check_out} shows that when the validity condition in \eqref{eq:check_out0} is satisfied, one can further refine the contract with the following value iteration:
\begin{equation}\label{eq:refine}
  {y^{\max }}[k + 1] = \Lambda \left( {{y^{\max }}[k]} \right).
\end{equation}
\begin{prop}
  \eqref{eq:refine} always converges when \eqref{eq:check_out0} is satisfied.
\end{prop}
\begin{proof}
 By the monotonicity of $\Lambda$ and the definition of $y^{\max }$, we have
 \begin{equation}
   \forall k=0,1,2,...,\;\; 0 \le y^{\max }[k + 1]\le y^{\max }[k].
 \end{equation}
 Then by the bounded convergence theorem, the value iteration converges.
\end{proof}

\section{Search for Assume-Guarantee Contract with Epigraph Method}\label{sec:epigraph}
In this section, we present the epigraph method that searches for an assume-guarantee contract that meets the validity condition. In particular, we show that the epigraph method can be viewed as an extension of the classic small gain theorem to network systems with nonlinear `gains'.
\subsection{Epigraph representation of the validity condition}
With Lemma \ref{lem:check_out}, the key problem now is to find a contract that meets the validity condition, which means to find $y^{\max}$ such that
\begin{equation}\label{eq:check_out1}
  \forall i=1,...,N,\quad \lambda_i(y_{\pazocal{N}_i}^{\max})\le y_i^{\max}.
\end{equation}
We propose an epigraph algorithm to search for such a $y^{\max}$. The main idea is to look at each $\lambda_i:\pazocal{Y}_{\pazocal{N}_i}\to\pazocal{Y}_i$ in Lemma \ref{lem:check_out}. The condition in \eqref{eq:check_out1} is equivalent to the following condition:

%Here we use a robust invariant set to bound $y_i$ given $y_{\pazocal{N}_i}^{\max}$. The check out condition in \eqref{eq:check_out0} becomes the following condition:
%\begin{equation}\label{eq:check_out}
%\begin{array}{c}
%  \forall \Sigma_i, \exists k_i(x_i,y_{\pazocal{N}_i},d_i),\pazocal{S}_i, s.t. \\
%  \forall x_i\in\pazocal{S}_i,\forall \left|y_{\pazocal{N}_i}\right|\le y_{\pazocal{N}_i}^{\max},\forall d_i\in\pazocal{D}_i, \\
%  x_i^+=f_i\left(x_i,y_{\pazocal{N}_i},k_i(x_i,y_{\pazocal{N}_i},d_i),d_i\right)\in\pazocal{S}_i,\\
%  \max\limits_{x_i\in\pazocal{S}_i}{\left|h_i(x_i)\right|}\le y_i^{\max}.
%\end{array}
%\end{equation}

\begin{equation}\label{eq:epi}
  [y_{\pazocal{N}_i}^{\max};y_i^{\max}]\in\epi(\lambda_i),
\end{equation}
where $\epi(\cdot)$ denotes the epigraph of a scalar function. Suppose the epigraph of each $\lambda_i$ is known, the search for an initial valid contract can be formulated as the following feasibility problem:
\begin{equation}\label{eq:search_initial_contract}
  \begin{aligned}
\mathop {\min }\limits_{{y^{\max }} \ge {\bf{0}}} \;&0\\
&\mathrm{s.t.}\;\forall i = 1,...,N, \left[y_{\pazocal{N}_i}^{\max};y_i^{\max}\right]\in\epi(\lambda_i).
\end{aligned}
\end{equation}
%Therefore, for a network system, given $\epi(\lambda_i)$, the following optimization solves for a contract that checks out:
%\begin{equation}\label{eq:contract_solve}
%  \begin{array}{l}
%\mathop {\min }\limits_{{y^{\max }}} \;0\;s.t.\forall i = 1,...,N,\\
%{\left[{\left( {y_{\pazocal{N}_i}^{\max }} \right)^\intercal},y_i^{\max }\right]^\intercal} \in \epi(\lambda_i)
%\end{array}
%\end{equation}
If $\epi(\lambda_i)$ is hard to get, one can replace $\epi(\lambda_i)$ in \eqref{eq:search_initial_contract} with its inner approximation and the optimization would still generate a valid contract if a solution is obtained. Once a valid contract is obtained, it can be further refined by the value iteration shown in \eqref{eq:refine}.
\begin{exmp}
Consider the two systems interconnection network shown in Fig. \ref{fig:two_interconnection}.
\begin{figure}[H]
    \centering
    \includegraphics[width=1.8in]{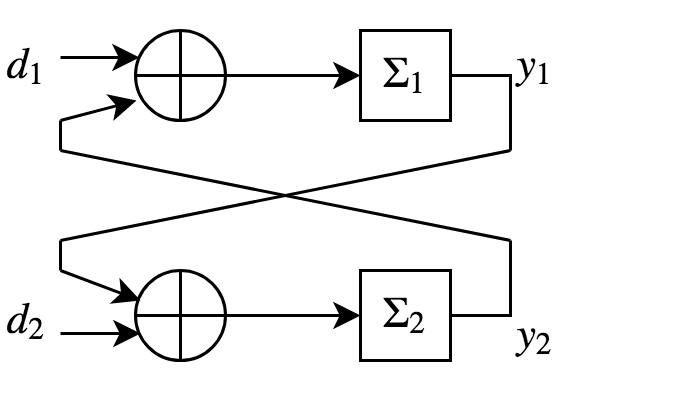}
    \caption{Two systems interconnection network}\label{fig:two_interconnection}
\end{figure}
Suppose that the two subsystems satisfy
  \begin{equation}\label{eq:inter_assume}
    \begin{array}{l}
{\left\| {{y_1}} \right\|_\infty } \le {\mu _1}{\left\| {{d_1}} \right\|_\infty } + {\nu _1}{\left\| {{y_2}} \right\|_\infty }\\
{\left\| {{y_2}} \right\|_\infty } \le {\mu _2}{\left\| {{d_2}} \right\|_\infty } + {\nu _2}{\left\| {{y_1}} \right\|_\infty }
\end{array}
  \end{equation}
In addition, the small gain condition is satisfied, i.e.,
\begin{equation}\label{eq:small_gain}
  \nu_1\cdot\nu_2<1.
\end{equation}
Then by small gain theorem, the interconnected network is stable and
\begin{equation}\label{eq:small_gain_res}
  \begin{aligned}
{\left\| {{y_1}} \right\|_\infty } &\le & \frac{{{\mu _1}}}{{1 - {\nu _1}{\nu _2}}}{\left\| {{d_1}} \right\|_\infty } &+ \frac{{{\mu _2}{\nu _1}}}{{1 - {\nu _1}{\nu _2}}}{\left\| {{d_2}} \right\|_\infty }\\
{\left\| {{y_2}} \right\|_\infty } &\le & \frac{{{\mu _1}{\nu _2}}}{{1 - {\nu _1}{\nu _2}}}{\left\| {{d_1}} \right\|_\infty } &+ \frac{{{\mu _2}}}{{1 - {\nu _1}{\nu _2}}}{\left\| {{d_2}} \right\|_\infty },
\end{aligned}
\end{equation}
see \cite{kim2017small} for detail. The same result can be obtained by considering the epigraph.
\begin{cor}
  Given \eqref{eq:inter_assume} and $\left\| {d_i} \right\|_\infty,i=1,2$, not both zero, there exists an assume-guarantee contract that guarantees \eqref{eq:small_gain_res} if $\nu_1\cdot\nu_2<1$.
\end{cor}
\begin{proof}
  Given \eqref{eq:inter_assume}, $\left\| {d_{1}} \right\|_\infty$, and $\left\| {d_{2}} \right\|_\infty$, $\lambda_{1,2}$ can be easily found to be
\begin{equation}
\begin{aligned}
  \lambda_1(\left\| {{y_2}} \right\|_\infty)&=&\mu_1 \left\| {{d_1}} \right\|_\infty &+ \nu_1 \left\| {{y_2}} \right\|_\infty \\
  \lambda_2(\left\| {{y_1}} \right\|_\infty)&=&\mu_2 \left\| {{d_2}} \right\|_\infty &+ \nu_2 \left\| {{y_1}} \right\|_\infty.
\end{aligned}
\end{equation}
\begin{figure}[H]
  \centering
  \includegraphics[width=0.48\columnwidth]{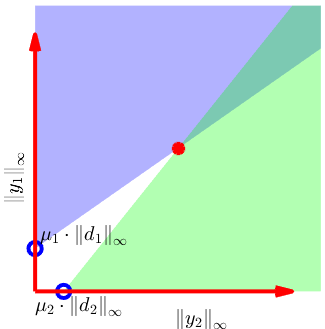}
  \caption{Epigraph of $\lambda_{1,2}$ for the interconnected system}\label{fig:epi_interconnection}
\end{figure}
The epigraph of $\lambda_{1,2}$ are shown in Fig. \ref{fig:epi_interconnection}, where the blue shade shows $\epi(\lambda_1)$ and the green shade shows $\epi(\lambda_2)$. A contract is valid if the point $[\left\| {{y_1}} \right\|_\infty, \left\| {{y_2}} \right\|_\infty]^\intercal$ lies within the intersection of the two epigraphs. When $\left\| {{d_1}} \right\|_\infty$ and $\left\| {{d_2}} \right\|_\infty$ are not both zero, the two epigraphs have a nonempty intersection if and only if $\nu_1\cdot\nu_2<1$. When the intersection is nonempty, the contract with the minimum $\left\|y_{1,2}\right\|_{\infty}$ is depicted as the red dot, which can be verified to be equal to the result in \eqref{eq:small_gain_res}.

\end{proof}
\end{exmp}
\begin{rem}
  The small gain theorem is a special case of the epigraph method, which can be extended to cases when $\lambda_i$ are nonlinear functions and when there are more than 2 interconnected subsystems.
\end{rem}

\subsection{Grid Sampling for epigraph approximation}

Next, we show a grid sampling approach to compute an inner-approximation of $\epi(\lambda_i)$.
For the simplicity of notation, we consider a scalar function $f:\mathbb{R}^n\to\mathbb{R}$, with input $x$ and output $y=f(x)$.

The epigraph of a function is not bounded since it is defined as the area above the function graph in $[x;f(x)]$ space, as shown in Fig. \ref{fig:epigraph}. Besides, the domain of $x$ may be unbounded as well.
\begin{figure}[H]
  \centering
\includegraphics[width=0.9\linewidth]{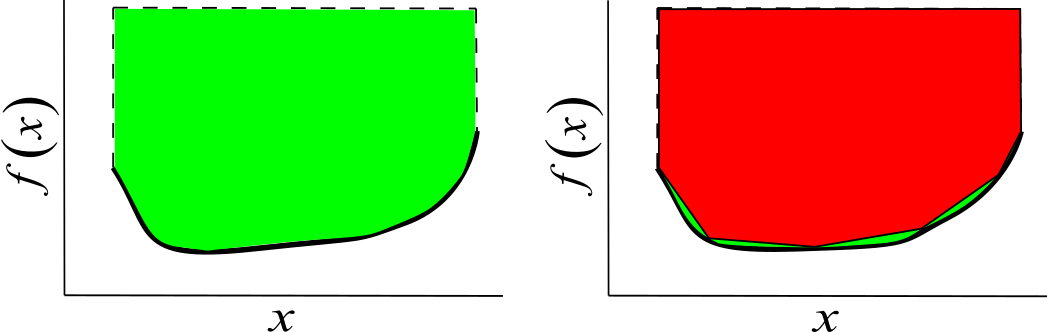}
\caption{Epigraph of a function and its polytopic approximation}\label{fig:epigraph}
\end{figure}
Therefore, to get a reasonable representation of the epigraph, we first need to fix the domain of $x$ to be a compact set $\pazocal{X}$ of interest, then pick a large constant $M$ such that $\forall x\in\pazocal{X}, f(x)<M$. Then we look for a cropped inner approximation of the epigraph.

First notice that by definition, when $f$ is a convex function, then its epigraph is a convex set. If one picks a finite set $S=\left\{x_1,x_2,...,x_n\right\}$ and evaluate the function at every point in $S$, then compute the convex hull of the point set $[x_1;f(x_1)],[x_2;f(x_2)],...,[x_n;f(x_n)]$, denoted as $H$, then $H$ is convex and $H\subseteq \epi(f)$. If for each $x_i$, we add $[x_i;M]$ to the point set, we get a cropped inner approximation of $\epi(f)$, as shown in the second figure in Fig. \ref{fig:epigraph}. Therefore, for a convex function, we can simply sample the input and use the convex hull of the sampled points with their function values as the approximation of $\epi(f)$.

When $f$ is not convex, a decomposition algorithm is developed to inner approximate $\epi(f)$ with a union of polytopes. The decomposition algorithm is omitted.
%an example is shown in Fig. \ref{fig:polytope_decomp}.
%\begin{figure}[H]
%  \centering
%  \includegraphics[width=0.48\columnwidth]{polytope_decomp.eps}
%  \caption{Approximation of $\epi(f)$ as a union of polytopes}\label{fig:polytope_decomp}
%\end{figure}
In this case, suppose $\epi(f)$ is approximated by ${\bigcup\limits_{j = 1}^M {{p_j}} }$, where $p_j$ are polytopes, then $[x;f(x)]\in\epi(f)$ is encoded with the following mixed integer constraint:
\begin{equation}\label{eq:mip}
[x;f\left( x \right)] \in \bigcup\limits_{j = 1}^M {{p_j} \Leftrightarrow } \left( \begin{array}{l}
\mathds{1}([x;f\left( x \right)] \in {p_j}) - {s_j} \ge 0,\\
{s_j} \in \left\{ {0,1} \right\},\sum\limits_{j = 1}^M {{s_j} = 1} ,
\end{array} \right)
\end{equation}
where $s_j$ are the binary variables and $\mathds{1}(\cdot)$ is the indicator function.
\section{Example application to power grid control}\label{sec:application}
In this section, we apply the proposed method on a microgrid control problem as an example to demonstrate the benefit of the method.
\subsection{Microgrid problem setup}
The microgrid control is an important network control application. There has been a lot of effort focusing on the stability, optimality, and safety of the network \cite{molzahn2017survey,zhao2014design,giani2009viking}. This paper is motivated by the need to improve the transient performance of the Optimal Power Flow (OPF) based controller studied in \cite{zhao2014design,mallada2017optimal}. Although the OPF controller achieves good asymptotic performance, it lacks guarantee for the transient performance. In particular, when sudden changes such as failure of a component or a short circuit at one of the nodes happen, drastic change on frequency should be avoided since it may lead to severe damage to the system and heavy economic loss.

Correct-by construction control synthesis is a good complement to the existing controller since it provides performance guarantee to the transient of the system and can work with any existing controller. However, application of correct-by-construction techniques such as robust control invariant sets on the microgrid and other network control problems has been difficult due to the high state dimension of the network systems. The assume-guarantee reasoning method proposed in this paper is a potential solution to this problem of scalability since it decomposes the large network system into small subsystems with bounded disturbances, which can be handled by existing computation tools for correct-by-construction synthesis. Since the grid network is typically sparse, i.e., a node is usually connected to only a few neighbors, the epigraph method proposed in Section \ref{sec:epigraph} has linear complexity, which makes the computation of robust invariant sets for a microgrid possible.

We consider the IEEE 9-bus test case, where the parameters are from the Power System Toolbox (PST) \cite{chow1992toolbox}, as shown in Fig. \ref{fig:9b3g}.
\begin{figure}[H]
  \centering
  \includegraphics[width=1.7in]{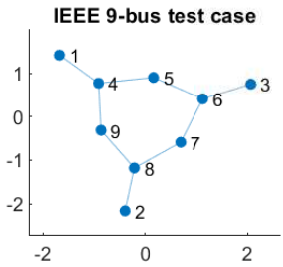}
  \caption{Network structure of the microgrid}\label{fig:9b3g}
\end{figure}
The generator buses are $\pazocal{G}=\left\{1,2,3\right\}$ and the load buses are $\pazocal{L}=\left\{4,5,6,7,8,9\right\}$. The dynamics of the micro-grid can be described by the following model \cite{mallada2017optimal}:
\begin{equation}\label{eq:grid_model}
\resizebox{.9\hsize}{!}{$
\begin{aligned}
  \dot{\theta}_i&=\omega_i,\\
  M_i\dot{\omega}_i&=P_i^{in}-D_i \omega_i - d_i-u_i -\sum\limits_{j\in\pazocal{N}_i}B_{ij}(\theta_i-\theta_j), i\in\pazocal{G}\\
  0&=P_i^{in}-D_i \omega_i - d_i-u_i -\sum\limits_{j\in\pazocal{N}_i}{B_{ij}(\theta_i-\theta_j)}, i\in\pazocal{L},
\end{aligned}
$}
\end{equation}
where $\theta_i$ and $\omega_i$ are the phase angle and frequency of the voltage at bus $i$, $P_i^{in}$ and $d_i$ are the input power and uncontrollable load at bus $i$, the sudden change of them is the main source of disturbance to the system. $u_i$ is the controllable load, which is used to regulate bus $i$. $\pazocal{G}$ and $\pazocal{L}$ represent the set of generator buses and the set of pure load buses. For a generator bus, $M_i$ is the inertia and $D_i$ is the ``damping coefficient''; for a load bus, there is zero inertia and $\omega_i$ is determined by an algebraic equation. A generator bus is modeled with 2 states ($x_i=[\theta_i,\omega_i]^\intercal$); and a load bus is modeled with 1 state ($x_i=\theta_i$) for a load bus. $B_{ij}$ represents the sensitivity of the power flow to phase variations, it is nonzero when bus $i$ and bus $j$ are neighbors. The output $y_i=\theta_i$ since the coupling between buses happen through $\theta_i$.

The control objective is to prevent large frequency deviation from a set value. However, since the coupling happens via the phase angle differences, in order to bound the frequency deviation, one need to bound phase angles as well. The approach we take is to compute a robust control invariant set (RCI) for each bus, which is robust against sudden changes in the input power and uncontrollable load and the coupling between neighboring buses. In addition, the frequency deviation bound is always satisfied inside the RCI.
\subsection{Search for RCI with epigraph algorithm}
For each bus, the RCI computation depends on the available input, bound on possible exogenous disturbance and bound on the phase angles of neighboring buses. Denote the invariant set, the input bound and exogenous disturbance bound of bus $i$ as $\pazocal{S}_i$, $\pazocal{U}_i$ and $\pazocal{D}_i$, respectively. $\pazocal{U}_i$ and $\pazocal{D}_i$ are determined by the environment assumption and are assumed to be given, while the bound on phase angle deviation of neighboring buses $\theta_{\pazocal{N}_i}^{\max}$ is given as the feedback assumption.

It should be emphasized that the epigraph method works with any method that can compute a robust invariant set given the disturbance bound. Therefore, the specific algorithm of RCI computation is not the focus of this paper. In particular, we used a robust optimization approach to compute the robust invariant set, which uses a polytope with fixed template as the representation of the RCI and iteratively solve for an RCI through robust optimization. See \cite{chen2018RCI} for detail.

Denote the RCI computation process as $\mathcal{F}$, which takes $\pazocal{U}_i$, $\pazocal{D}_i$, the dynamics $\Sigma_i$ and $\theta_{\pazocal{N}_i}^{\max}$ as input and generates $\pazocal{S}_i$:
\begin{equation}
  \pazocal{S}_i=\mathcal{F}(\pazocal{U}_i,\pazocal{D}_i,\Sigma_i,\theta_{\pazocal{N}_i}^{\max}).
\end{equation}
\begin{defn}
  $\mathcal{F}$ is \textit{monotonic} w.r.t. $\theta^{\max}$ if for any fixed $\pazocal{U}$, $\pazocal{D}$ and $\Sigma$, given $\theta^{\max,1}\ge \theta^{\max,2}\ge \mathbf{0}$, let $\pazocal{S}^i = \mathcal{F}\left( \pazocal{U},\pazocal{D},\Sigma,\theta^{\max,i} \right)$, then $\pazocal{S}^2\subseteq \pazocal{S}^1$. The inequality is defined element-wise.
\end{defn}
\begin{prop}
There exists a $\mathcal{F}$ that is monotonic w.r.t. $\theta^{\max}$.
\end{prop}

\begin{proof}
$\theta^{\max,1}\ge \theta^{\max,2}$ implies that the uncertainty set for $\pazocal{S}^1$ contains the uncertainty set for $\pazocal{S}^2$, so $\pazocal{S}^1$ is also robust control invariant under $\theta^{\max,2}$. Therefore, picking $\pazocal{S}^2 = \pazocal{S}^1$ completes the proof.
\end{proof}
\begin{ass}\label{ass:F_monotone}
  The algorithm $\mathcal{F}$ for computing the robust control invariant set is monotonic.
\end{ass}
Let
\begin{equation}\label{eq:lambda_grid}
  \lambda_i(\theta_{\pazocal{N}_i}^{\max})=\max\limits_{x_i\in\pazocal{S}_i}{\left|\theta_i\right|}.
\end{equation}
By Assumption \ref{ass:F_monotone}, $\lambda_i$ is clearly monotonic. The evaluation of $\lambda_i$ is done in two steps. First, with $\theta_{\pazocal{N}_i}^{\max}$ fixed, $\mathcal{F}$ is called to compute an RCI $\pazocal{S}_i$, then $\theta_i^{\max}$ is obtained through \eqref{eq:lambda_grid}.

Then the inner approximation of $\epi(\lambda_i)$ is computed for each bus with the grid sampling algorithm, Fig. \ref{fig:epi_grid_example} shows two computed epigraph as examples:
\begin{figure}[H]
  \centering
  \includegraphics[width=0.85\columnwidth]{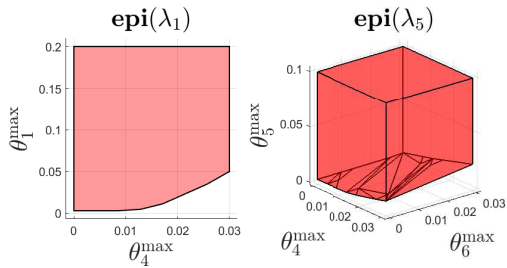}
  \caption{Inner approximations of $\epi(\lambda_1)$ and $\epi(\lambda_5)$}\label{fig:epi_grid_example}
\end{figure}
Since some of the epigraphs are not convex, a mixed integer programming as formulated in \eqref{eq:mip} is solved. Once a valid assume-guarantee constraint is obtained, robust invariant sets for each subsystem can be obtained via $\mathcal{F}$.

 Fig. \ref{fig:RCI} shows the robust invariant sets for the generator buses under the assume-guarantee contract.

\begin{figure}[H]
  \centering
  \includegraphics[width=0.9\columnwidth]{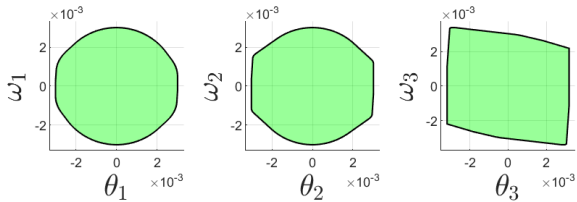}
  \caption{Robust control invariant sets for the generator buses}\label{fig:RCI}
\end{figure}

\subsection{Simulation result}
% We assume that for each bus, the bound for controllable load is 1 and the maximal uncontrolled load fluctuation is 0.5.

For each bus, the computed robust control invariant set is then used to construct a control barrier function (CBF), which acts as a supervisor. The CBF supervisory control was first proposed in \cite{ames2014control}, where the authors proposed a Quadratic Programming framework that keeps the system safe with minimum intervention. The robust optimization algorithm generates an RCI with a polytopic representation: $\left\{ {x\in\mathbb{R}^n|Px \le q} \right\}$, where $P$ is a constant $L\times n$ matrix and $q\in\mathbb{R}^L_{>0}$. Note that the origin is always contained in the interior of the RCI. The CBF is defined as
\begin{equation}\label{eq:CBF}
  b\left( x \right) = \mathop {\min }\limits_k \frac{{{q_k} - {P_k}x}}{{{q_k}}}
\end{equation}
The supervisory control is implemented with the following quadratic programming:
\begin{equation}\label{eq:qp}
  \begin{aligned}
u*=\mathop {\arg\min }\limits_u \;&{\left\| {u - {u_0}} \right\|^2}\\
&s.t.\quad\dot b(x,u) + \kappa b(x) \ge 0,
\end{aligned}
\end{equation}
where $u_0$ is the control input of a student controller and $\kappa$ is a positive constant. In this case the primal-dual controller introduced in \cite{mallada2017optimal} is used as the student controller. The second line of \eqref{eq:qp} is called the CBF condition. It can be shown that when the CBF condition is satisfied, $x$ stays inside the RCI, see \cite{ames2017control} for detail. The quadratic programming in \eqref{eq:qp} will leave $u_0$ unchanged if $u_0$ satisfies the CBF condition and use minimum intervention when it doesn't. \eqref{eq:qp} is always feasible for a $\kappa$ large enough if $\left\{ {x\in\mathbb{R}^n|Px \le q} \right\}$ is an RCI.
\begin{figure}[H]
  \centering
  \includegraphics[width=1.0\columnwidth]{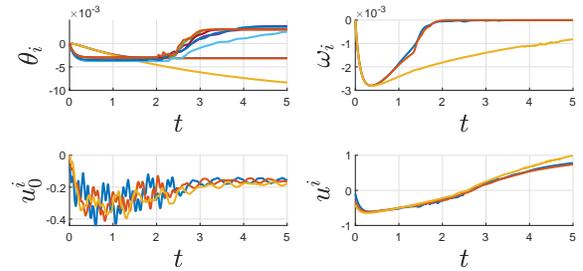}
  \caption{Simulation with CBF as supervisor}\label{fig:sim_CBF1}
\end{figure}
Fig. \ref{fig:sim_CBF1} shows the result of simulation when CBF is acting as a supervisor. The bound on frequency deviation is set at $5\times10^{-3}rad/s$ and was never breached.
\begin{figure}[H]
  \centering
  \includegraphics[width=1.0\columnwidth]{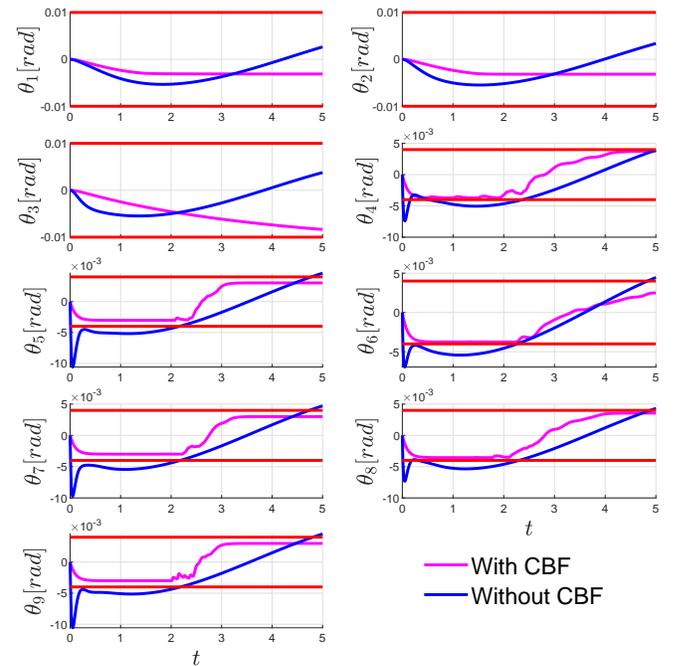}
  \caption{phase angle plot with and without CBF as supervisor}\label{fig:sim_CBF2}
\end{figure}

Fig. \ref{fig:sim_CBF2} shows the values of $\theta_i$ with and without the CBF supervisor. Under the CBF supervisory controller, all the $\theta_i$s are within their respective bound determined by the contract; on the other hand, without CBF, there is no guarantee that the phase angles stay within bounds under $u_0$.

%\begin{figure}[H]
%  \centering
%  \includegraphics[width=3in]{grid_9b3g_sim_plot1_u0.eps}
%  \caption{phase angle plot without CBF as supervisor}\label{fig:sim_no_CBF2}
%\end{figure}
\section{Conclusion}\label{sec:conclusion}
We propose an assume-guarantee reasoning based method to compute a robust control invariant set for a network system. The coupling between subsystems are treated as bounded disturbances and is handled with an assume-guarantee contract. We show that an assume-guarantee contract satisfying the validity condition guarantees robust set invariance for a network system and can be further refined with value iteration. When such a valid contract is not known, an epigraph algorithm is proposed to search for a valid contract, which enjoys linear complexity when the network is sparse. It is shown that the epigraph algorithm can be viewed as an extension of the classic small gain theorem to network systems with nonlinear `gains'. The proposed method is demonstrated with a microgrid control example. The epigraph algorithm is able to find a valid contract which leads to robust invariant sets for each subsystem in the network. Then control barrier functions are constructed based on the robust invariant sets which then act as supervisors to keep the states inside their respective invariant sets under exogenous disturbances and coupling between the subsystems.
\balance
\bibliographystyle{myieeetran}
\bibliography{Grid_bib}
\end{document}